\documentclass[conference]{IEEEtran}
\IEEEoverridecommandlockouts

\usepackage{cite}
\usepackage{amsmath,amssymb,amsfonts}
\usepackage{physics}
\usepackage{algorithm}
\usepackage{algorithmic}

\usepackage{stmaryrd}

\usepackage{graphicx}
\usepackage{placeins}
\usepackage{textcomp}
\usepackage{xcolor}
\def\BibTeX{{\rm B\kern-.05em{\sc i\kern-.025em b}\kern-.08em
    T\kern-.1667em\lower.7ex\hbox{E}\kern-.125emX}}

\newtheorem{theorem}{Theorem}
\newtheorem{proposition}{Proposition}
\newtheorem{corollary}{Corollary}
\newtheorem{definition}{Definition}

\newtheorem{lemma}{Lemma}
\newenvironment{proof}{\noindent\textbf{Proof:}}{\hfill$\square$\par}

\DeclareMathOperator\supp{supp}

\newcommand{\F}{\mathbb{F}}
\newcommand{\Er}{\mathcal{E}}
\newcommand{\cP}{\mathcal{P}}
\newcommand{\cD}{\mathcal{D}}
\newcommand{\cR}{\mathcal{R}}

\begin{document}

\title{Quantum Maxwell Erasure Decoder for qLDPC codes
\thanks{AL acknowledges support from Agence Nationale de la Recherche through project ANR-22-PETQ-0006 of Plan France 2030.}
}

\author{\IEEEauthorblockN{Bruno Costa Alves Freire}
\IEEEauthorblockA{Inria Paris, Pasqal \\
bruno.costa-alves-freire@inria.fr}
\and
\IEEEauthorblockN{François-Marie Le Régent}
\IEEEauthorblockA{Pasqal \\
francois-marie.le-regent@pasqal.com}
\and
\IEEEauthorblockN{Anthony Leverrier}
\IEEEauthorblockA{Inria Paris \\
anthony.leverrier@inria.fr}
}

\maketitle

\begin{abstract}
We introduce a quantum Maxwell erasure decoder for CSS quantum low-density parity-check (qLDPC) codes that extends peeling with bounded guessing. Guesses are tracked symbolically and can be eliminated by restrictive checks, giving a tunable tradeoff between complexity and performance via a guessing budget: an unconstrained budget recovers Maximum-Likelihood (ML) performance, while a constant budget yields linear-time decoding and approximates ML. We provide theoretical guarantees on asymptotic performance and demonstrate strong performance on bivariate bicycle and quantum Tanner codes. \end{abstract}


\section{Introduction}

Quantum low-density parity-check (qLDPC) codes form a promising approach to low-overhead quantum fault tolerance~\cite{gottesman2014fault,bravyi2024high}. Realizing this promise requires fast decoding, and message-passing methods are attractive because they can scale nearly linearly with the blocklength.
However, direct adaptations of classical belief propagation---highly successful for classical LDPC codes~\cite{richardson2008modern}---face additional challenges for stabilizer codes linked to the presence of short cycles in the Tanner graph and degeneracy~\cite{poulin2008iterative}. 

While circuit-level noise models are ultimately relevant for fault-tolerant quantum computing, we focus here on the quantum erasure channel~\cite{bennett1997capacities,grassl1997codes}, where the locations of erased qubits are revealed. In classical LDPC history, the binary erasure channel served as a clean testbed that led to new tools (e.g., stopping sets and density evolution) and informed practical code design~\cite{1003839,luby2002efficient,richardson2008modern}.
The erasure setting is similarly attractive in the quantum case: maximum-likelihood (ML) can be formulated as a linear algebra problem on the erased positions and Gaussian elimination (with a cubic complexity) provides a concrete benchmark for best possible performance~\cite{delfosse2020linear}. Moreover, erasures are also physically relevant (loss, leakage, or flagged qubits) and can arise through ``erasure conversion'' techniques~\cite{wu2022erasure,scholl2023erasure}. 

The simplest iterative erasure decoder is \emph{peeling}: whenever a check is adjacent to exactly one erased variable, the constraint immediately determines that variable, and the algorithm propagates this information. 
Peeling coincides with the belief propagation (BP) decoder on the erasure channel and runs in linear time, but stops on a \emph{stopping set} (SS), i.e., a subset of erased positions whose induced residual Tanner graph contains no degree-1 check nodes. A crucial quantum feature is the presence of small SS induced by low-weight stabilizers. Understanding and exploiting the structure of such SS is therefore central to quantum erasure decoding. For instance, pruning and the VH decoder~\cite{Connolly_2024} are particularly effective for hypergraph product codes~\cite{Tillich_2014}. 
Recent constructions---including lifted product codes~\cite{panteleev2022asymptotically,panteleev2021degenerate}, bivariate bicycle (BB) codes~\cite{kovalev2013quantum,bravyi2024high} or quantum Tanner codes~\cite{leverrier2022quantum,leverrier2025small}---reinforce the need  for designing general-purpose and scalable decoders beyond the hypergraph product setting.

Besides VH and pruning, several general-purpose quantum erasure decoders have recently been proposed. The \emph{cluster decoder} post-processes the residual SS by decomposing it into clusters and solving small clusters exactly.
Increasing the cluster-size cutoff interpolates between fast decoding and ML performance~\cite{cluster_yao_gokduman_pfister}. 
Another approach adapts belief propagation with guided decimation (BPGD) to erasures, progressively fixing variables to help convergence~\cite{gokduman2024erasure}. These methods trade performance and complexity. Recent work also shows that degeneracy-aware BP-style decoders can achieve near-capacity performance on the quantum erasure channel for a broad class of stabilizer codes~\cite{kuo2024degenerate}.

We propose a quantum variant of the Maxwell decoder~\cite{maxwell_classic} that also interpolates between peeling and ML. Starting from peeling, whenever decoding stalls on a stopping set, we guess an erased variable and resume peeling. One can either branch over many assignments or track guesses symbolically. As peeling proceeds, some formal relations may be identified, leading to ``guess reimbursement''. A single parameter $G_{\max}$ bounds the number of active guesses and yields an explicit performance-complexity tradeoff.
Our contributions are a quantum Maxwell decoder for CSS qLDPC codes; a complexity analysis showing linear scaling at fixed $G_{\max}$; an asymptotic small-$\epsilon$ result relating the number of guesses to the failure exponent; and simulations on BB and quantum Tanner codes, including comparisons with the cluster decoder.

Section~\ref{sec:algo_description} presents the decoder, Section~\ref{sec:theory} analyzes complexity and the $\epsilon \to 0$ regime performance, and Section~\ref{sec:numerics} reports our numerical results.

\section{Quantum Maxwell decoder} \label{sec:algo_description}

We consider the quantum erasure channel, where an erasure set $\mathcal{E}_0 \subseteq[n]$ of qubits is revealed to the decoder. We restrict to CSS stabilizer codes specified by parity-check matrices
$H_X\in\F_2^{m_X\times  n}$ and $H_Z\in\F_2^{m_Z\times n}$ satisfying $H_XH_Z^T=0$.
Writing a Pauli error in symplectic form as $(e_X,e_Z)\in\F_2^n\times\F_2^n$ with
$\supp(e_X)\cup\supp(e_Z)\subseteq \mathcal \Er_0$, the measured syndromes satisfy
\begin{equation}\label{eq:css_syndromes}
\sigma_Z = H_Z e_X,\qquad \sigma_X = H_X e_Z.
\end{equation}
Thus erasure decoding for a CSS code reduces to two binary linear problems on the \emph{same} erasure pattern:
find $w^X,w^Z\in\F_2^n$ supported in $\mathcal{E}_0$ such that
$H_Z w^X=\sigma_Z$ and $H_X w^Z=\sigma_X$.

The main subproblem is as follows: given a binary parity-check matrix $H\in\F_2^{m\times n}$, syndrome $\sigma\in\F_2^m$, and erasure set $\mathcal{E}_0$, we seek a solution $w\in\F_2^n$ such that $\supp(w)\subseteq\mathcal{E}_0$ and $Hw=\sigma$. We denote by $\mathcal{E}$ the current residual erasure set during the decoding and by $T_{\mathcal E}(H)$ the residual Tanner graph induced by the erased variables. A check is \emph{dangling} if it has residual degree $1$ in $T_{\mathcal E}(H)$; peeling repeatedly resolves dangling checks in linear time. Peeling stops on a \textit{stopping set} (SS), i.e., when there are no dangling checks.

The Maxwell decoder augments peeling with guesses: whenever peeling gets stuck, we \emph{guess} one erased variable and resume peeling. Conceptually, $g$ guesses create $2^g$ branches, one per assignment of the guessed bits. To avoid explicit branching, one can represent all messages as \emph{affine forms} over $\F_2$ in a set of pivot variables $\mathbf x=(x_p)_{p\in\mathcal{P}}$, $\mathcal{P}\subset [n]$: each variable assignment $w_j$ and each check right-hand side $s_i$ is stored as $a_0+\sum_{p\in\mathcal P} a_p x_p$ with $a_0,a_p\in\F_2$. A newly guessed erased variable becomes a new pivot: $w_g\leftarrow x_{g}$, $\mathcal{P}\gets\mathcal{P}\cup\{g\}$.

Guesses are validated through \emph{restrictive checks}.
During peeling, a check may become degree $0$ in the residual graph while its cumulative syndrome is nonzero;
this yields a constraint $s_i(\mathbf x)=0$ on pivots.
We enforce such constraints by \emph{pivot demotion}: pick one pivot appearing in $s_i$ and rewrite it as an affine function of the others,
then substitute everywhere. This will reduce $|\mathcal P|$ and effectively \emph{reimburse} guesses.
One may for instance demote the most recently introduced pivot appearing in the constraint, and we consider this policy later in the complexity analysis.

Algorithm~\ref{algo:maxwellpeel} gives pseudocode for the symbolic Maxwell routine.

\begin{algorithm}
\caption{\textsc{MaxwellPeel}$(H,\sigma,\mathcal{E}_0,G_{\max})$ (symbolic)}
\label{algo:maxwellpeel}
\begin{algorithmic}[1]
\REQUIRE Parity-check matrix $H$, syndrome $\sigma$, erasure set $\mathcal{E}_0 \subseteq[n]$, guess budget $G_{\max}$
\ENSURE \textbf{Failure} or affine-form assignments $w(\mathbf x)$ supported on $\mathcal{E}_0$
\STATE Initialize erased set $\mathcal{E} \gets \mathcal{E}_0$
\STATE Initialize variable corrections $w_j \gets 0$ for all $j\in[n]$
\STATE Initialize cumulative syndromes $s_i \gets \sigma_i$ for all checks $c_i$
\STATE Initialize pivot set $\mathcal P \gets \varnothing$
\STATE Initialize dangling checks $\mathcal D$ and restrictive checks $\mathcal R$
\STATE $\mathcal D \gets \{\,c_i : \deg_{\mathcal E}(c_i)=1\,\}$ \hfill
      $\mathcal R \gets \{\,c_i : \deg_{\mathcal E}(c_i)=0 \ \wedge\ s_i\neq 0\,\}$

\WHILE{$\mathcal E \neq \varnothing$}

  \WHILE{$\mathcal D \neq \varnothing$}
    \STATE Pop $c_i$ from $\mathcal D$
    \IF{$\deg_{\mathcal E}(c_i)=1$} 
      \STATE Let $v_j$ be the unique erased neighbor of $c_i$
      \STATE Set $w_j \gets s_i$
      \STATE Remove $v_j$ from $\mathcal E$
      \STATE Propagate $w_j$ to all neighboring checks $c_{i'}\in N(v_j)$:
             update $s_{i'}$ and the residual degrees
      \STATE Update $\mathcal D$ and $\mathcal R$
    \ENDIF
  \ENDWHILE

  \WHILE{$\mathcal R \neq \varnothing$}
    \STATE Pop $c_i$ from $\mathcal R$ \hfill \{constraint $s_i(\mathbf x)=0$\}
    \STATE Choose a pivot $x\in \mathrm{supp}(s_i)$ (e.g., the most recent one) and solve for $x$
    \STATE Substitute $x$ everywhere in affected $w_j$ and $s_{i'}$
    \STATE Update $\mathcal P \gets \mathcal P \setminus \{x\}$ and update $\mathcal R$
  \ENDWHILE

  \IF{$\mathcal E \neq \varnothing$}
    \IF{$|\mathcal P| = G_{\max}$}
      \RETURN \textbf{Failure}
    \ENDIF
    \STATE Pick a pivot variable node $v_g \in \mathcal E$ \hfill \{pivot selection rule\}
    \STATE Introduce new pivot $x_{|\mathcal P|+1}$ and set $w_g \gets x_{|\mathcal P|+1}$
    \STATE Remove $v_g$ from $\mathcal E$, propagate to neighboring checks, update $\mathcal D,\mathcal R$
    \STATE $\mathcal P \gets \mathcal P \cup \{x_{|\mathcal P|+1}\}$
  \ENDIF

\ENDWHILE

\RETURN $w(\mathbf x)$
\end{algorithmic}
\end{algorithm}

We obtain the quantum Maxwell decoder by applying Algorithm~\ref{algo:maxwellpeel} independently to the two CSS components:
\begin{align}
w^X(\mathbf x)&\leftarrow \textsc{MaxwellPeel}(H_Z,\sigma_Z,\Er_0,G_{\max}),\\
w^Z(\mathbf y)&\leftarrow \textsc{MaxwellPeel}(H_X,\sigma_X,\Er_0,G_{\max}).
\end{align}
Evaluating the affine forms at an assignment of the pivots yields explicit corrections. For online decoding, it suffices to output any valid evaluation (e.g., $\mathbf x=\mathbf 0$) after all restrictive checks have been applied. For benchmarking, one needs to verify if the corresponding correction differs from the true error by an element of the stabilizer group.

Two orthogonal choices can be added to the decoder without changing the core routine:
(i) \emph{pruning} before guessing, as in pruned peeling/VH~\cite{Connolly_2024}, which can remove stabilizer-induced stopping sets;
(ii) a \emph{pivot selection rule} (uniformly at random, or a local score based on how many dangling checks would be created by guessing a variable).
We evaluate these options in Appendices \ref{app:pruning} and \ref{app:strategies}.

\section{Theoretical properties}
\label{sec:theory}
In this section we establish and prove theoretical results on the runtime complexity and decoding performance of the quantum Maxwell decoder described in Section \ref{sec:algo_description}. 
Throughout, we analyze one binary instance $(H,\sigma,\Er_0)$; the CSS decoder applies the same
analysis independently to $(H_Z,\sigma_Z,\Er_0)$ and $(H_X,\sigma_X,\Er_0)$.

\subsection{Notation and data structure model}
\label{subsec:notation}

Fix a binary parity-check matrix $H\in\mathbb{F}_2^{m\times n}$ with Tanner graph having variable
nodes $V=[n]$ and check nodes $C=[m]$. For a node $u$ we write $N(u)$ for its neighbors. The input
to \texttt{MAXWELLPEEL} is a syndrome vector $\sigma\in\mathbb{F}_2^m$ and an initial erasure set
$\Er_0\subseteq[n]$.

The algorithm maintains a \emph{residual} erased set $\Er \subseteq \Er_0$ (initialized as $\Er\leftarrow \Er_0$).
For a check $c_i\in C$, its residual degree is
\begin{equation}
\deg_{\Er}(c_i)\triangleq |N(c_i)\cap \Er|.
\end{equation}
The set of \emph{dangling checks} and \emph{restrictive checks} at any time are
\begin{align}
\cD &\triangleq \{\,c_i\in C:\deg_E(c_i)=1\,\}, \\
\cR &\triangleq \{\,c_i\in C:\deg_E(c_i)=0\ \wedge\ s_i(x)\neq 0\,\}.
\end{align}
Here $s_i(x)$ is the algorithm's maintained \emph{cumulative syndrome} (an affine form over the
current pivot variables), and $s_i(x)\neq 0$ means ``not the identically-zero affine form.''

Guesses are tracked via an \emph{active pivot set} denoted $\cP=\{x_{p_1},\dots,x_{p_g}\}$ with $g=|\cP|\le G_{\max}$.
Each variable correction $w_j(x)$ and each cumulative syndrome $s_i(x)$ is stored as an affine form
over $\mathbb{F}_2$ in the pivots $\cP$:
\begin{equation}
a_0 + \sum_{p\in P} a_p x_p,\qquad a_0,a_p\in\mathbb{F}_2.
\end{equation}

We represent affine forms as vectors of length at most $G_{\max}+1$ (one constant bit plus one coefficient bit per pivot).
Hence affine addition and affine substitution both cost $O(G_{\max})$ bit operations. 

We rely on the following pivot ranking policy (PRP): whenever a restrictive check is processed, the algorithm demotes  the \emph{most recently introduced} pivot appearing in the constraint. Equivalently, pivots are ordered by creation time and
the newest pivot in $\mathrm{supp}(s_i)$ is eliminated.

\subsection{Complexity analysis} \label{sec:complexity}

We assume $H$ is $(d_v,d_c)$-LDPC, meaning that columns and rows have weight at most $d_v$ and $d_c$, respectively, with $d_v, d_c = O(1)$. Let $e\triangleq |\Er_0|$.  

\begin{theorem}[Runtime of symbolic \texttt{MAXWELLPEEL} under PRP]
\label{thm:runtime}
Algorithm~\ref{algo:maxwellpeel} runs in $O\left(e d_v d_c G_{\max}^2\right)$ bit operations. In particular, for fixed $G_{\max}$ and bounded degrees, the runtime is $O(e)$.
\end{theorem}

\begin{proof}
We decompose the cost into: (i) degree/bookkeeping, (ii) affine propagation during peeling/guessing,
and (iii) substitutions during pivot demotions.

\paragraph*{(i) Degree bookkeeping and maintaining $\cD$ and $\cR$}
Whenever a variable $v_j$ is removed from $\Er$ (peeled or guessed), we scan its $d_v$ neighboring checks
and decrement their residual degrees. Since each erased variable is removed exactly once, this contributes $O(e\,d_v)$ time.

\paragraph*{(ii) Peeling/guess propagation of affine forms}
When $v_j$ is removed, the algorithm propagates $w_j(x)$ to each adjacent check $c_{i'}\in N(v_j)$ by
$s_{i'}(x)\leftarrow s_{i'}(x)+w_j(x)$. Each affine addition costs $O(|\cP|)\le O(G_{\max})$ and there are at most $e\,d_v$ such updates. Hence the total propagation cost is $O(e\,d_v\,G_{\max})$. 

\paragraph*{(iii) Bounding total substitution work from pivot demotions}
A pivot demotion eliminates one pivot $x_p$ by expressing it as an affine form in older pivots and
substituting it into all affected affine forms. Under PRP, substitutions never introduce pivots newer
than the eliminated pivot.

For each erased variable $v_j\in \Er_0$, let $\tau(j)$ be the time at which $v_j$ is removed from the
residual set $\Er$ and its affine correction $w_j(x)$ is \emph{last assigned} by either peeling or guessing.

\begin{lemma}
\label{lem:no-future-pivots}
Under PRP, after time $\tau(j)$, the affine form $w_j(x)$ can only be modified by demotions of pivots
that already existed at time $\tau(j)$. In particular, $w_j(x)$ is updated by at most $G_{\max}$ pivot
demotions.
\end{lemma}

\begin{proof}
After $\tau(j)$, the algorithm never reassigns $w_j$ (the variable is no longer in $\Er$); the only
possible changes come from substituting eliminated pivots. Under PRP, each demotion rewrites the
\emph{newest} pivot in a restrictive constraint as a function of strictly older pivots, so substitutions
cannot create occurrences of pivots introduced after $\tau(j)$. Therefore only pivots that existed at
$\tau(j)$ can ever affect $w_j$, and there are at most $G_{\max}$ such pivots, each
eliminated at most once.
\end{proof}

By Lemma~\ref{lem:no-future-pivots}, across all $e$ erased variables, the total number of updates to
variable affine forms $\{w_j\}_{j\in \Er_0}$ caused by demotions is at most $e\,G_{\max}$, and each such
update costs $O(G_{\max})$, giving $O(e\,G_{\max}^2)$.

Next, we bound substitution work on the check affine forms $\{s_i\}$. Each check $c_i$ has at most
$d_c$ neighboring erased variables, and each such neighbor's $w_j$ changes at most $G_{\max}$ times due
to Lemma~\ref{lem:no-future-pivots}. Thus each check's $s_i$ is updated at most $d_c\,G_{\max}$ times
due to demotions. The number of distinct checks adjacent to $\Er_0$ is at most $e\,d_v$, hence the total number of demotion-driven updates to checks
is $O(e\,d_v\,d_c\,G_{\max})$, each costing $O(G_{\max})$, for a total check-substitution cost
$O(e\,d_v\,d_c\,G_{\max}^2)$.

This establishes the theorem.
\end{proof}

In appendix \ref{app:strategies}, we discuss the impact on the complexity of implementing the score-based pivot selection rule mentioned in Section \ref{sec:algo_description}, providing here only its definition.

\begin{definition}(Guess score) \label{def:gs}
    Define the \emph{guess score} of an erased variable $v$ as the number of neighboring check nodes of degree exactly $2$. The score-based guessing strategy consists in picking the variable with the highest score to be the new pivot. 
\end{definition}

\subsection{Performance in the asymptotic regime}

\begin{definition}[Stopping sets and stopping distance]
\label{def:stopping}
A nonempty set $S\subseteq[n]$ is a \emph{stopping set} for $H$ if every check node adjacent to $S$ has at
least two neighbors in $S$. The \emph{stopping distance} of $H$ is
\[
s(H)\triangleq \min\{|S|: S\neq\varnothing \text{ is a stopping set for }H\}.
\]
For a CSS code, define $s\triangleq \min\{s(H_X),s(H_Z)\}$.
\end{definition}

For $w \leq n$, let $A^{\mathrm{SS/NTLO}}(w)$ be the number of $w$-subsets of $[n]$ that are respectively stopping sets or the support of nontrivial logical operators of the code. We are interested in the case where $A^{\mathrm{SS}}(w)$ is strictly larger than $A^{\mathrm{NTLO}}(w)$ because they correspond to erasures for which the Maxwell decoder can fail more often than the ML decoder.

\begin{definition}[Distribution gaps]
For $t \leq n$, we define $W_0(t) \triangleq \{ 0 < w \leq t : A^{\mathrm{SS}}(w) > A^{\mathrm{NTLO}}(w) \}$ and the \emph{distribution gap} to be $\gamma(t) \triangleq |W_0(t)|$.
\end{definition}

For a deterministic non-persistent decoder whose success/failure depends only on the erasure set, the logical failure probability on the quantum erasure channel with rate $\epsilon$ can be written as
\[
p_L(\epsilon)=\sum_{E\subseteq[n]} \epsilon^{|E|}(1-\epsilon)^{n-|E|}\,\mathbf{1}[E\in\mathcal{F}],
\]
where $\mathcal{F}$ is the set of erasure patterns on which the decoder fails. Let $\mathcal{F}_{\rm ML}$ denote the set of erasure patterns that are information-theoretically uncorrectable (equivalently, contain the support of a nontrivial logical operator), and let $\mathcal{F}_{\rm QM}(G_{\max})$ be the set of erasure patterns on which the Maxwell decoder fails due to the guess budget. We denote $p_L^{\rm ML}$ and $P_L^{{\rm QM}(G_{\max})}$ the corresponding logical failure probabilities.

\begin{theorem}[Budget exhaustion implies large erasure]
\label{thm:failure-large-strong}
Consider an erasure set $\Er_0 \notin \mathcal{F}_{\mathrm{ML}}$, such that $|\Er_0| \le t$, for some $t \in [n]$. If $G_{\max} \ge \gamma(t)$, then $\Er_0 \notin \mathcal{F}_{\rm QM}(G_{\max})$. 
\end{theorem}

Intuitively, under the assumption that $\Er_0 \notin \mathcal{F}_{\rm ML}$, every state where \texttt{PEELING} gets stuck corresponds to an ML-correctable set, and each guess strictly decreases the size of the erased set, so the number of guesses required is bounded by the number of possible sizes, that is $\gamma(t)$.

\begin{proof}
    The proof relies on the following statement, to be proven by induction: 
    \begin{align*}
        \text{If } \Er_0 \notin\; & \mathcal{F}_{\rm ML} \text{ and } |\Er_0| < \min W_{G}(t),\\& \text{ then } \Er_0 \notin \mathcal{F}_{\rm QM}(G), \tag{Cond($G$)}
    \end{align*}
    where we define 
    $$W_G(t) \triangleq \begin{cases}
        W_0(t) \setminus \left\{\mqty{W_0(t) \text{'s } G \\\text{ smallest elements}}\right\}, \text{ if } G < \gamma(t)\\
        \{t+1\}, \text{ if } G\ge\gamma(t).
    \end{cases}$$
    
    To prove the base case Cond($0$), notice that if $\Er_0 \notin \mathcal{F}_{\rm ML}$ with $|\Er_0| < \min W_0(t)$, that means $\Er_0$ cannot contain a SS, because such structures of compatible size would correspond to ML failure cases, and thus $\Er_0 \notin \mathcal{F}_{\rm QM}(0)$, i.e., the \texttt{MAXWELLPEEL} algorithm succeeds without using a single guess. 

    Next, for the inductive step, assume Cond($G-1$) for some $G>0$ and consider an erasure set $\Er_0 \notin \mathcal{F}_{\rm ML}$ with $|\Er_0| < \min W_G(t)$. The peeling stage of the \texttt{MAXWELLPEEL} algorithm reduces $\Er_0$ to $\Er$, another SS. Hence, $\mathcal{E} \in W_0(t)$ and because $|\mathcal{E}| < \min W_G(t)$, it is of size at most $\min W_{G-1}(t)$, since there is no SS of intermediary size. Then, through the application of a guess, if need be, one gets $\Er' \notin \mathcal{F}_{\rm ML}$ and $|\Er'| < |\Er| \le \min W_{G-1}(t)$. It follows from Cond($G-1$) that $\Er' \notin \mathcal{F}_{\rm QM}(G-1)$, and thus, since $\Er'$ is obtained from $\Er$ by using at most $1$ guess, $\Er \notin \mathcal{F}_{\rm QM}(G)$, establishing Cond($G$). 

    To complete the proof, notice that if $G \ge \gamma(t) = |W_0(t)|$, we have $W_G(t) = \{t+1\}$, and thus $\min W_G(t) = t+1$. Hence, for $\Er_0 \notin \mathcal{F}_{\rm ML}$, $|\Er_0| \le t < \min W_{G_{\max}}(t)$ implies through Cond($G_{\max}$) that $\Er_0 \notin \mathcal{F}_{\rm QM}(G_{\max})$. 
\end{proof}

Define
\[
\Delta_{G_{\max}}(\epsilon)\triangleq p^{\rm QM(G_{\max})}_L(\epsilon)-p^{\rm ML}_L(\epsilon),
\]
so that $\Delta_{G_{\max}}$ counts the additional failures of Maxwell on ML-correctable patterns.

\begin{corollary}[Divisibility of the gap polynomial]
\label{cor:gap-gamma}
With $\Delta_{G_{\max}}$ defined above, for $G_{\max} \ge \gamma(t)$ we have
\[
\Delta_{G_{\max}}(\epsilon)=O\!\left(\epsilon^{\,t+1}\right)\qquad (\epsilon\to 0).
\]
\end{corollary}

\begin{proof}
$\Delta_{G_{\max}}$ sums only over patterns $\Er_0$ where ML succeeds but Maxwell returns \textsf{Failure}.
By Theorem~\ref{thm:failure-large-strong}, no such patterns exist with $|\Er_0|<t+1$. Hence the lowest
power of $\epsilon$ appearing in $\Delta_{G_{\max}}$ is at least $t+1$.
\end{proof}

Applying the straightforward bounds $\gamma(d-1) \le d-s$ and $\gamma(d) \le d-s+1$ yields the following.

\begin{corollary}[Matching the ML exponent for CSS codes]
\label{cor:match-ml}
Let $d$ be the CSS code distance  and
$s=\min\{s(H_X),s(H_Z)\}$. Then $p^{\rm ML}_L(\epsilon)=\Theta(\epsilon^d)$.
 If $G_{\max}\ge d-s+1$, then $p^{\rm QM(G_{\max})}_L(\epsilon)\sim p^{\rm ML}_L(\epsilon)$ as $\epsilon\to 0$.
\end{corollary}

In practice, these bounds may be very loose, meaning that each guess in the budget is usually worth more than one unit in the exponent of the $p_L$ curve. Improvements such as the guess reimbursement mechanism and sophisticated guessing strategies also help make each guess be worth even more.

\section{Numerical results}
\label{sec:numerics}

We report results for an explicit-branching implementation of the Maxwell routine, which is equivalent to running Algorithm~\ref{algo:maxwellpeel} but explores up to $2^{G_{\max}}$ assignments in parallel. Since we restrict to $G_{\max}\leq 6$ (at most 64 branches), branching is practical here. The symbolic implementation analyzed in Section~\ref{sec:theory} removes this exponential factor and is better suited for larger guessing budgets.
In addition, we adopted the score-based guessing strategy from Definition \ref{def:gs}.

We consider the erasure channel, where each qubit is erased independently with probability $\epsilon$. Decoding is considered successful if and only if the algorithm returns a unique solution, modulo the stabilizer group. 

We evaluate the decoders for two families of qLDPC codes: two BB codes from \cite{bravyi2024high} and two Quantum Tanner codes \cite{leverrier2022quantum,leverrier2025small}. 

Figure \ref{fig:QM_benchmark} compares the performance of the quantum Maxwell decoder for various guessing budgets against peeling, ML (Gaussian elimination on the erased positions) and the cluster decoder from \cite{cluster_yao_gokduman_pfister}. For each value of $\epsilon$, we sample up to $10^7$ random erasure patterns (or $10^6$ for the cluster decoder), and stop after 100 failures.

\begin{figure*}[htb]
    \centering
    \includegraphics[width=0.9\linewidth]{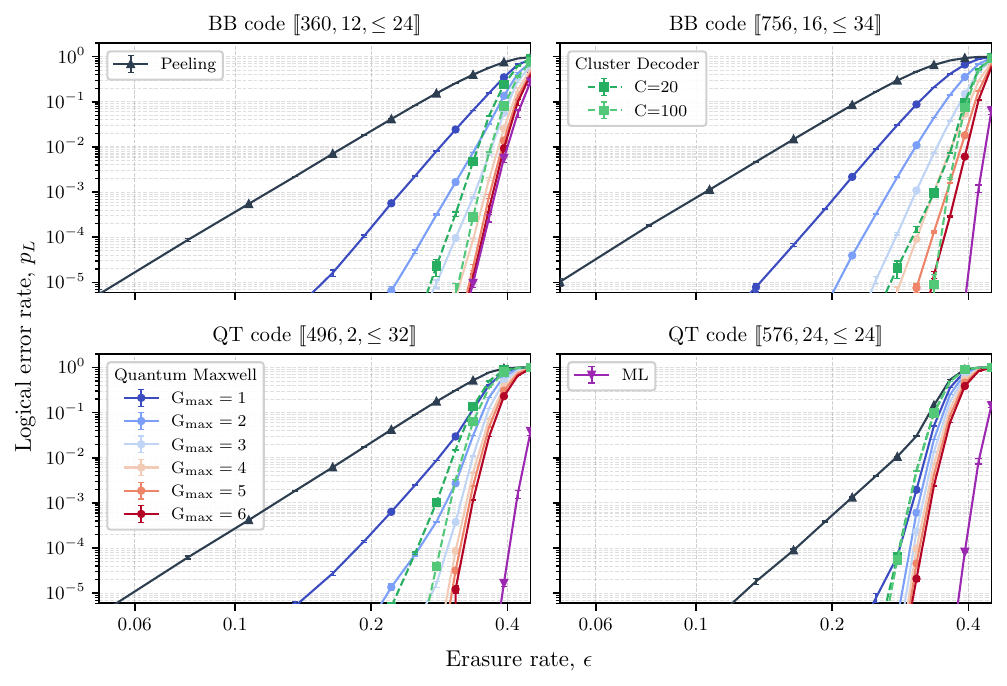}
    \caption{
    Logical error rate $p_L$ versus erasure rate $\epsilon$ for two BB codes (top) and two quantum Tanner codes (bottom). We compare peeling (black), ML decoding (purple), the quantum Maxwell decoder with $1 \leq G_{\max} \leq 6$ (with score-based pivot selection) and the cluster decoder with cluster cutoff $C \in \{20, 100\}$~\cite{cluster_yao_gokduman_pfister}.}
    \label{fig:QM_benchmark}
\end{figure*}

We observe that the quantum Maxwell (QM) decoder interpolates the performances of peeling and ML as we vary the guessing budget $G_{\max}$. For the $\llbracket 360,12, \leq 24\rrbracket$ BB code, QM seemingly matches ML performance with up to $6$ guesses. On the remaining codes, a visible performance gap remains in the waterfall region. Due to the high parameters of the codes, probing the error floor region is difficult (in terms of sampling complexity) as we increase the guessing budget. 

The QM and cluster decoders exhibit different tradeoffs. The logical error rate drops faster for the QM decoder in the waterfall region, but the cluster decoder catches up more quickly as the erasure rate decreases. While both decoders trade complexity for performance, it is difficult to establish a direct equivalence between the tunable parameters (maximum cluster size $C$, for cluster decoder, and guessing budget $G_{\max}$, for quantum Maxwell) of the two algorithms.

\section{Conclusion}
We have presented a quantum variant of the Maxwell decoder for CSS quantum LDPC codes on the erasure channel. The decoder provides an explicit performance-complexity tradeoff through a single parameter $G_{\max}$ that bounds the number of guesses,  allowing it to either reproduce ML performance, or approximate it with linear complexity.

We proved that the symbolic Maxwell routine runs in $O(|\Er_0|\,d_v\,d_c\,G^2_{\max})$ bit operations. In particular, for bounded degrees and constant $G_{\max}$, the runtime is linear in the number of erasures, making the approach practically viable for large codes. In the asymptotic regime $\epsilon \to 0$, we showed that $G_{\max} \ge d - s + 1$ guesses are enough to match ML performance at the leading order (same failure exponent), where $d$ and $s$ denote the minimum and stopping distances respectively. The bound is not always tight, and multiple improvements on the algorithm may make the guessing budget worth more. For instance, integration with preprocessing methods such as pruning may yield additional gains for some code families. 

Numerical simulations on bivariate bicycle and quantum Tanner codes demonstrate that the quantum Maxwell decoder interpolates smoothly between peeling and ML, achieving near-optimal logical error rates with modest guessing budgets, while comparisons with the cluster decoder show competitive performance. 

In future work, the decoder can be extended to general stabilizer codes. Finally, in light of recent developments such as \cite{ye2025beamsearchdecoderquantum}, the adaptation of the symbolic quantum Maxwell to Pauli noise could be investigated, and further evaluated under circuit-level noise models.

\bibliographystyle{IEEEtran}
\bibliography{biblio}

\appendices

In this appendix, we discuss possible optimizations of the quantum Maxwell decoder. In Section~\ref{app:pruning}, we first study the impact of an additional pruning step and that of enlarging the set of low-weight generators to help peeling. In Section~\ref{app:strategies}, we consider the impact of optimizing the choice of the bit that are guessed.

\section{Pruning and redundant checks}
\label{app:pruning}

A simple improvement to peeling-based erasure decoders for qLDPC codes is to add a \emph{pruning} procedure~\cite{Connolly_2024}. The idea is to note that the most frequent stopping sets that the decoder will encounter are generators of the qLDPC code, and that it is simple to address them. Since the goal is to recover the true error up to a stabilizer element, one can in fact choose to add a generator to the error. This means that if a full generator is erased, then one can pick an element of its support and fix the corresponding value of the error for that element. It is possible to apply this procedure for all generators, but also combinations of several generators, as was investigated in \cite{Connolly_2024}. We define depth-$D$ pruning to correspond to pruning stabilizers consisting of combinations of up to $D$ generators. In particular, $D=0$ means no pruning and $D=1$ means we prune erased generators whose entire support is erased. In practice, $D>1$ does not improve the performance significantly and we don't consider it here. 
Formally, depth-1 pruning is as follows:
If a generator $g$ has $\mathrm{supp} \,  g \subseteq \mathcal{E}$, choose a variable $j \in \mathrm{supp} \, g$, set $w_j \leftarrow 0$ (gauge fix), and remove $j$ from $\mathcal{E}$. This breaks the stopping set induced by the generator without consuming a Maxwell guess. We iterate this until no-fully erased generator remains.
Figure \ref{fig:pruning} compares the Maxwell decoder with and without this extra pruning procedure (bottom vs top) and shows a modest improvement of performance when combining guessing with pruning. 

Next, we exploited a feature of certain quantum Tanner codes that admit many redundant generators of lowest weight. In particular, the quantum Tanner code $\llbracket 576, 24, \leq 24\rrbracket$ is naturally defined with generators of weight 9 that correspond to support of codewords of two classical $[6,3,3]$ codes~\cite{leverrier2025small}. Since these local codes admit 4 codewords of weight 3, the quantum Tanner code naturally admits many redundant generators of weight 9. 
Each of the two $[6,3,3]$ component codes has 4 weight-3 codewords, giving $4\times 4=16$ parity-checks of weight 9. Dividing by 9 to correct for overcounting (since the 16 redundant stabilizer elements replace the 9 operators from the original construction), one arrives at 1024 distinct stabilizers of weight 9. 
We form two regularized parity-check matrices that contain this full set of 1024 generators, yielding row weight 9 and column weight 8 for both $H_X$ and $H_Z$.
Running the Maxwell decoder on these regularized parity-check matrices turns out to significantly improve the performance, especially when combined with pruning, as can be observed on the right column of Figure \ref{fig:pruning}. That adding redundant parity-checks helps peeling-based decoders should not be surprising because additional parity-checks give more opportunities to obtain dangling checks.

\begin{figure*}[htb]
    \centering
    \includegraphics[width=.9\linewidth]{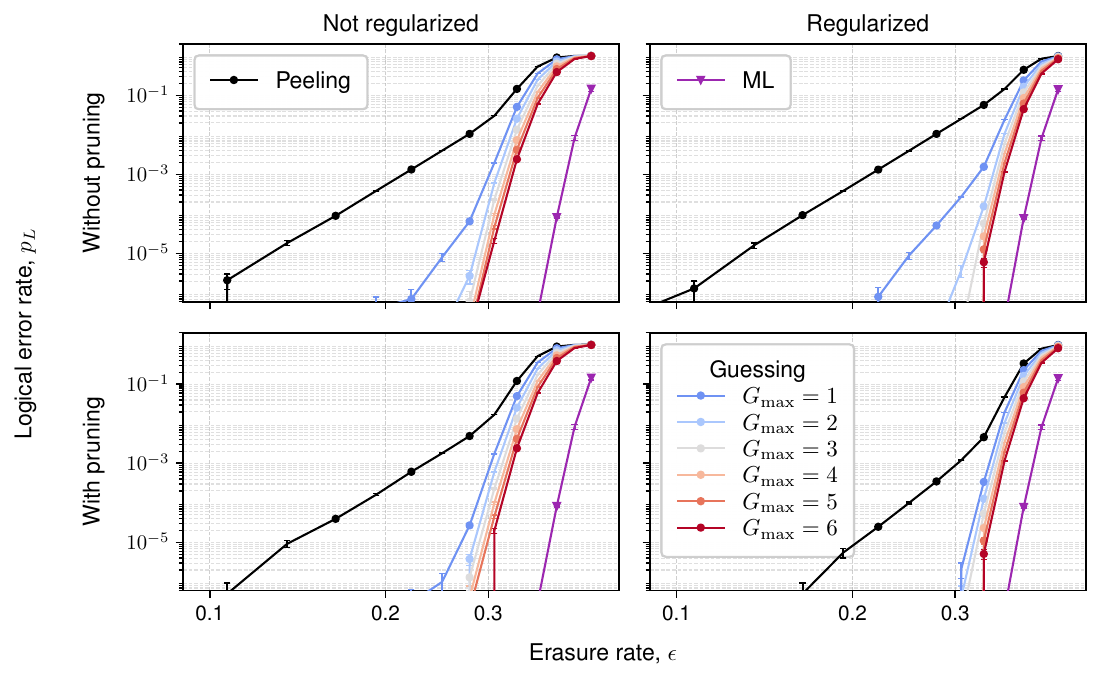}
    \caption{Impact of pruning and parity-check regularization on the $[[576,24,\leq 24]]$ quantum Tanner code. Rows: pruning off/on. Columns: original vs regularized parity-check matrices. Curves show peeling, ML, and quantum Maxwell for $1\leq G_{\max} \leq 6$.}
    \label{fig:pruning}
\end{figure*}

\section{Comparing guessing strategies}
\label{app:strategies}

The quantum Maxwell decoder presented in Algorithm \ref{algo:maxwellpeel} does not necessarily specify how to choose the pivot variable node. A possibility is to pick it at random among the erased set, but better strategies may exist. 
For instance, one can choose the guess score of Definition \ref{def:gs}, i.e., the number of neighboring checks of residual degree 2.
We compare these two strategies in Figure \ref{fig:guess_strat_plot}. It reports both $D=0$ (no pruning) and $D=1$ (depth-1 pruning), showing that score-based guessing reduces the low-$\epsilon$ error floor in both cases.

\begin{figure*}
    \centering
    \includegraphics[width=.9\linewidth]{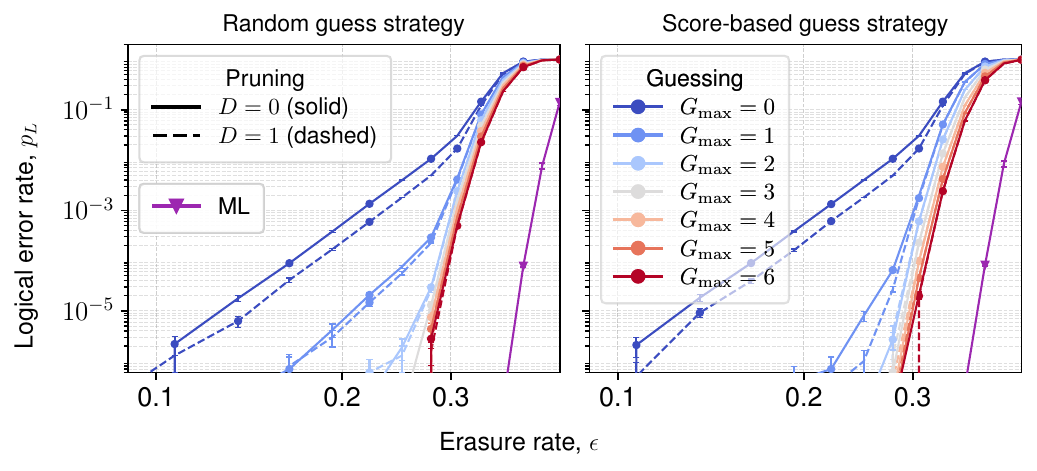}
    \caption{Random vs score-based pivot selection on the $[[576,24,\leq 24]]$ quantum Tanner code. Solid/dashed curves indicate pruning depth $D \in \{0,1\}$. Score-based selection significantly lowers the error floor at small $\epsilon$ for modest $G_{\max}$.
    }
    \label{fig:guess_strat_plot}
\end{figure*}

Moreover, we demonstrate that implementing the aforementioned guessing strategy does not incur in an increased complexity to the \texttt{MAXWELLPEEL} algorithm. 

\begin{proposition}
    Implementing the guessing strategy from Definition \ref{def:gs} does not increase the asymptotic time complexity of algorithm \texttt{MAXWELLPEEL}. 
\end{proposition}

\begin{proof}
    The key observation to establish this result is the fact that, because we assume $(d_v, d_c)$-LDPC codes, the guessing score is by definition a bounded quantity, $\mathrm{GS}(v_i) \in \{0, \dots, d_v\}$ for any variable node $v_i \in \Er$. With this in mind, one can use buckets to store the variables attaining each one of the possible scores. 
    
    First, one can initialize the guessing score for all erased variables in the beginning of the algorithm with cost $O(|\Er|\, d_v\, d_c)$, and initialize the buckets structure. Next, the scores have to be updated at each step of peeling. This can be done with a cost of at most $d_c$ degree updates for each variable removed from the erasure, leading to a cost $O(|\Er|\, d_v\, d_c)$ in the peeling phase to keep the bucket structure consistent. 

    The final step to verify is the cost of electing the new pivot, which requires finding a variable with maximum score. Because the buckets are kept up to date during the peeling phase, and there is only a constant amount of them, finding the largest score with a nonempty bucket has a cost $O(d_v)$, and thus the pivot selection also does not incur in an augmented complexity. 
\end{proof}

\end{document}